\documentclass[article,aps,pra,11pt,tightenlines]{revtex4-2}
\usepackage[a4paper, total={7in, 10in}]{geometry}
\usepackage{graphicx,epsfig}
\usepackage{amsmath}
\usepackage{amsfonts}
\usepackage{braket}
\usepackage{amsthm}
\usepackage{array}
\usepackage[dvipsnames]{xcolor}
\usepackage{algpseudocode} 
\usepackage{algorithm} 

\usepackage{overpic}
\usepackage{hyperref}
\usepackage{tabularx}
\usepackage{multirow}
\usepackage{array}
\newcolumntype{P}[1]{>{\centering\arraybackslash}p{#1}}

\newtheorem{theorem}{Theorem}

\theoremstyle{definition}
\newtheorem{definition}[theorem]{Definition}

\DeclareMathOperator{\Tr}{Tr}

\DeclareMathOperator*{\argmax}{\arg\!\max}
\DeclareMathOperator{\EX}{\mathbb{E}}

\newcommand{\tr}[1]{\mathrm{Tr}\left(#1\right)} 
\def\<{\langle}
\def\>{\rangle}

\begin{document}



\title{Quantum contextual bandits and \\ recommender systems for quantum data}
\author{Shrigyan Brahmachari$^1$, Josep Lumbreras$^1$, Marco Tomamichel$^{1,2}$}
\affiliation{$^1$Centre for Quantum  Technologies,  National University of Singapore, Singapore}
\affiliation{$^2$Department of Electrical and Computer Engineering, Faculty of Engineering, National University of Singapore, Singapore}

\date{\today}

\begin{abstract}
We study a recommender system for quantum data using the linear contextual bandit framework. In each round, a learner receives an observable (the context) and has to recommend from a finite set of unknown quantum states (the actions) which one to measure. The learner has the goal of maximizing the reward in each round, that is the outcome of the measurement on the unknown state. Using this model we formulate the low energy quantum state recommendation problem where the context is a Hamiltonian and the goal is to recommend the state with the lowest energy. For this task, we study two families of contexts: the Ising model and a generalized cluster model. We observe that if we interpret the actions as different phases of the models then the recommendation is done by classifying the correct phase of the given Hamiltonian and the strategy can be interpreted as an online quantum phase classifier.
\end{abstract}

\maketitle


\section{Introduction}
Recommender systems are a class of online reinforcement learning algorithms that interact sequentially with an environment suggesting relevant items to a user. During the last decade, there has been an increasing interest in online recommendation techniques due to the importance of advertisement recommendation for e-commerce websites or the rise of movies and music streaming platforms~\cite{gomezNetflix,dragone2019deriving}. Among different settings for recommender systems, in this work, we focus on the contextual bandit framework applied to the recommendation of quantum data. The contextual bandit problem is a variant of the multi-armed bandit problem where a learner at each round receives a context and given a set of actions (also called actions) has to decide the best action using the context information. After selecting an action the learner will receive a reward and for the next rounds, they will use the previous information of contexts and rewards in order to make their future choices. As in the classical multi-armed bandit problem, the learner has to find a balance between exploration and exploitation; exploration refers to trying different actions in order to eventually learn the ones with the highest reward and exploitation refers to selecting the actions that apparently will give the highest reward immediately. For a comprehensive review of bandit algorithms, we refer to the book by Lattimore and Szepesv\'ari~\cite{lattimore2020bandit}. Some real-life applications~\cite{bouneffouf2020survey} of bandit include clinical trials~\cite{durand2018contextual}, dynamic pricing~\cite{dynamical}, advertisement recommendation~\cite{advertisement} or online recommender systems~\cite{contextual_new_article_recommendation,mcinerney2018explore}. As an example, in~\cite{contextual_new_article_recommendation} a news article recommender system was considered where the context is the user features, the actions are the articles to recommend and the reward is modeled as a binary outcome indicating that the user clicks or not on the recommended article. 

Quantum algorithms for the classical multi-armed bandit problem have been studied for the settings of best-arm identification~\cite{quantumbandits,quantumbandits2}, exploration-exploitation with stochastic environments~\cite{wan2022quantum} (uncorrelated and linear correlated actions) and adversarial environments~\cite{cho2022quantum}. Also, a quantum neural network approach was considered in~\cite{hu2019training} for a simple best-arm identification problem. A quantum algorithm for a classical recommender system was considered in~\cite{kerenidis2017quantum} claiming an exponential speedup over known classical algorithms but later in~\cite{tang2019quantum} it was proven that the price of the speedup comes from the assumptions of the quantum state preparation part and argued that under related classical assumptions a classical algorithm can also achieve the speedup. There are other more general reinforcement learning frameworks beyond bandits where actions affect the rewards in the long term such as Markov decision process. The quantum generalization of this framework has been considered in~\cite{barry11,ying2021}, and although our model of study falls into their class we can derive more concrete results since we study a specific setting.

We are interested in studying a recommender system for quantum data that is modeled by a set of unknown quantum processes- which is called the \textit{environment}, and a set of tasks to perform using these quantum processes- which is called a \textit{context set}. A learner interacts sequentially with the environment receiving at each round a task from the context set and then choosing the best quantum process to perform this task. For example, we could model the environments as a set of noisy quantum computers, the context set as a set of different quantum algorithms, and then at each round, the learner is given a quantum algorithm to run and their goal is to recommend the best quantum computer to do this task. We note that this model exemplifies the bandit exploration-exploitation trade-off since the learner has to try (explore) the different quantum computers in order to decide the best one but at the same time has to choose the best one (exploitation) to perform the task. This trade-off is interesting in a practical scenario because it captures settings where online decisions are important and or they have some associated cost that makes the learner always try to perform optimally. In our example, one could think that using a quantum computer costs money for the learner, so at each stage, they always want to select the ones that will output the best solutions.
\begin{figure}[H]
    \centering
    \includegraphics[scale = 0.5]{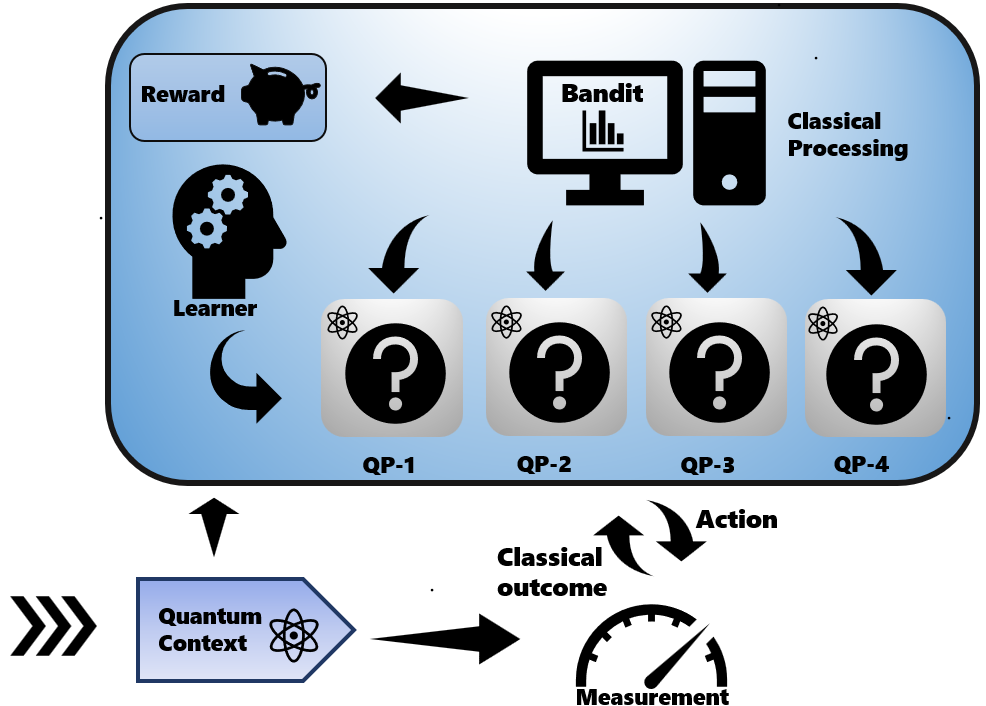}
    \caption{Sketch of a recommender system for quantum data. The learner receives sequentially quantum contexts and feed them to the classical processing system. The context is also fed to the measurement system. The classical processing system uses the information about the context to pick one of the quantum processes (no information regarding these processes are known besides from measurements). The chosen quantum process is applied to the measurement system, and the measurement outcome is fed to the classical processing and is added to the cumulative reward.}
    \label{fig:QCB}
\end{figure}

In our work, we extend the setting considered in~\cite{lumbreras21bandit} where they studied the exploration-exploration trade-off of learning properties of quantum states. In our model, the environment is a set of unknown quantum states, the context set is a (finite or infinite) set  of observables and at each round the learner receives an observable and has to perform a measurement on one of the unknown quantum states (the recommendation) aiming to maximize its outcome. We define this problem as the \textit{quantum contextual bandit} (QCB) and we note that it falls into the class of linear contextual bandits~\cite{abe2003reinforcement,aurer,chu}. The QCB is the basic framework where we formulate our recommender system for quantum data. We use as a figure of merit the regret, which is the cumulative sum of the difference between the expected outcome of the best and selected action at each round. Finding a strategy that minimizes regret implies finding the mentioned balance between exploration-exploitation of the different actions. As a concrete recommendation task captured by the QCB model, we consider the \textit{low energy quantum state recommendation problem}. In this problem, at each round, the learner receives a quantum Hamiltonian and has to recommend from the environment the state with the lowest energy. The ground state preparation problem is an important ingredient of NISQ algorithms~\cite{kishor_noisy2022} and our model could be useful in order to implement an online recommendation algorithm that helps the learner choose the best ansatz for their energy minimization task when they have multiple problems to solve. One of the advantages of using bandit algorithms for this task is that they do not need to reconstruct the whole $d$-dimensional state, just the relevant part for the recommendation which depends on the structure of the context set. In order to do that we combine a Grahm-Schmidt procedure with classical linear bandit strategies. This allows our algorithm to store low-dimensional approximations of the unknown quantum states without prior knowledge of the context set. We also perform some numerical studies of the scaling of the regret for the cases where the context set is an Ising model and a generalized cluster model studied in~\cite{verresen2017one}. For these models, we propose unknown actions for the algorithm corresponding to ground states located at different phases, and then for each context received by the algorithm we associate each action to a different phase and we reproduce a ground state phase diagram. We observe that the recommendation of the algorithm is done  approximately by classifying the different phases of the studied models and we are able to clearly distinguish them in the phase diagram.

The rest of the paper is organized as follows: in Section \ref{sec:model}, we establish the mathematical model for the quantum contextual bandit, and then define the notation used throughout the paper; in the next section, Section \ref{sec:lowerbound} we prove the lower bound on a performance metric (expected regret, which we define in Section \ref{sec:model}) over all possible algorithms.
In Section \ref{sec:algorithm} we review the linear Upper Confidence Bound algorithm. In Section \ref{sec:lowenergy} we describe the low-energy recommendation system and adapt the \textsf{LinUCB} algorithm to this setting. We illustrate the efficiency of the algorithm through simulations of different context sets.

\section{The model}\label{sec:model}
First, we introduce some notation in order to define our model and present our results. We define $[T] = \lbrace 1,...,T \rbrace$ for $T\in\mathbb{N}$. Let $\mathcal{S}_d = \lbrace \rho \in \mathbb{\mathbb{C}}^{d \times d}: \rho \geq 0 \wedge \Tr ( \rho ) = 1 \rbrace$ denote the set of positive semi-definite operators with unit trace, i.e \emph{quantum states} that act on a $d$-dimensional Hilbert space $\mathbb{C}^d$. Moreover, \emph{observables} are Hermitian operators acting on $\mathbb{C}^d$, collected in the set $\mathcal{O}_d = \lbrace O\in \mathbb{C}^{d\times d} : O^{\dagger} = O \rbrace $. We denote real $d$-dimensional column vectors as $\mathbf{v}$ and the inner product of two of them $\mathbf{u},\mathbf{v}\in\mathbb{R}^d$ as $\mathbf{u}^\top \mathbf{v}$ where $\mathbf{u}^\top$ denotes the transpose of $\mathbf{u}$ that is a row vector. We use $\| \cdot \|_2$ in order to denote the 2-norm of a real vector. For a $n$-qubit system with Hilbert space dimension $d=2^n$, we denote $X_i,Y_i$, and $Z_i$ the $x,y,z$ Pauli operators acting on the $i$-th qubit ($1\leq i \leq n$). A Pauli observable can be expressed as the $n$-fold tensor product of the $2 \times 2$ Pauli matrices, i.e it is an element of the set $\left\lbrace I,X,Y,Z \right\rbrace^{\otimes n}/ I_{4^n \times 4^n}$. Note that there are $4^n - 1$ such observables and each of them are orthogonal, and therefore form a basis, which alludes to as the \textit{Pauli basis} henceforth.

The definition of our model takes some of the conventions used for the multi-armed quantum bandit (MAQB) problem~\cite{lumbreras21bandit}.

\begin{definition}[Quantum contextual bandit] Let $d\in\mathbb{N}$. A $d$-dimensional \textit{quantum contextual bandit} is given by a set of observables $\mathcal{C} = \lbrace O_c \rbrace_{c\in\Omega_{\mathcal{C}}}\subseteq \mathcal{O}_d$ that we call the \textit{context set}, $(\Omega_{\mathcal{C}},\Sigma_{\mathcal{C}})$ is a measurable space and $\Sigma_{\mathcal{C}}$ is a $\sigma$-algebra of subsets of $\Omega_{\mathcal{C}}$. The bandit is in an \textit{environment}, a finite set of quantum states  $\gamma = \lbrace \rho_1,\rho_2,\cdots , \rho_k \rbrace \subset \mathcal{S}_d $, that it is unknown. The quantum contextual bandit problem is characterized by the tuple $(\mathcal{C},\gamma)$.
\end{definition}

 Given the environment $\gamma$ such that $ |\gamma | = k$ we define the \textit{action set} $\mathcal{A} = \lbrace 1,...,k \rbrace$ as the set of indices that label the quantum states $\rho_i\in\gamma$ in the environment. For every observable $O_c \in\mathcal{C}$ the spectral decomposition is given by 
\begin{align}
O_c = \sum_{i=1}^{d_c} \lambda_{c,i}\Pi_{c,i}, 
\end{align}
where $\lambda_{c,i} \in \mathbb{R}$ denote the $d_c \leq d$ distinct eigenvalues of $O_c$ and $\Pi_{c,i}$ are the orthogonal projectors on the respective eigenspaces.
For each action $a\in\mathcal{A}$ we define the reward distribution with outcome $R \in \mathbb{R}$ as the conditional probability distribution associated of performing a measurement using $O_c$ on $\rho_a$ given by Born’s rule
\begin{align}\label{eq:reward_distribution}
Pr\left[ R = r | A = a , O= O_c \right] = P_{\rho_a}(r|a,c) = \begin{cases}
\Tr ( \rho_a \Pi_{c,i} ) \text{ if } r= \lambda_{c,i}, \\
0 \text{ else. }
\end{cases}
\end{align}

 With the above definitions, we can explain the learning process. The learner interacts sequentially with the QCB over $T$ rounds such that for every round $t\in [T]$:
 
 \begin{enumerate}
 \item The learner receives a context $O_{c_t} \in \mathcal{C}$ from some (possibly unknown) probability measure \newline $P_\mathcal{C}:\Sigma \rightarrow [0,1]$ over the set $\Omega_\mathcal{\mathcal{C}}$. 
 \item Using the previous information of received contexts, actions played, and observed rewards the learner chooses an action $A_t\in\mathcal{A}$.
 \item The learner uses the context $O_{c_t}$ and performs a measurement on the unknown quantum state $\rho_{A_t}$ and receives a reward $R_t$ sampled according to the probability distribution~\eqref{eq:reward_distribution}. 
 \end{enumerate}
 
 We use the index $c_t \in [m]$ to denote the observable $O_{c_t}$ received at round $t\in[T]$. The strategy of the learner is given by a set of (conditional) probability distributions $\pi = \lbrace \pi_t \rbrace_{t\in\mathbb{N}}$ (policy) on the action index set $[k]$ of the form
\begin{align}
\pi_t (a_t|a_1,r_1,c_1,...,a_{t-1},r_{t-1},c_{t-1}, c_t ),
\end{align}
defined for all valid combinations of actions, rewards, and contexts $(a_1,r_1,c_1,...,a_{t-1},r_{t-1},c_{t-1} )$ up to time $t-1$.
Then, if we run the policy $\pi$ on the environment $\gamma$ over $T\in\mathbb{N}$ rounds, we can define a joint probability distribution over
the set of actions, rewards, and contexts as
\begin{align}\label{eq:reward_prob}
P_{\gamma,\mathcal{C},\pi} (a_1,X_1,C_1,...,a_T,X_T,C_T)  =  \int_{C_T} \int_{X_T}\cdots\int_{C_1} \int_{X_1}  \prod_{t=1}^T \pi_t (a_t|a_1,r_1,c_1,...,a_{t-1},r_{t-1},c_{t-1} ) \times \nonumber \\
\times P_{\mathcal{C}} (dc_1) P_{\rho_{a_1}} (d r_1 |a_1,c_1 )\cdots P_{\mathcal{C}} (dc_T) P_{\rho_{a_T}} (d r_T |a_T,c_T ). 
\end{align}
Thus, the conditioned expected value of reward $R_t$ is given by
\begin{align}
\EX_{\gamma,\mathcal{C},\pi} [ R_{t} | A_t = a , O_{c_t} = O_c ] = \Tr ( \rho_a O_c ),
\end{align}
where $\EX_{\gamma,\mathcal{C},\pi} $ denotes the expectation value over the probability distribution~\eqref{eq:reward_prob}. The goal of the learner is to maximize its expected cumulative reward $\sum_{t=1}^T \EX_{\gamma,\mathcal{C},\pi} \left[ R_t \right]$ or equivalently minimizing the \textit{cumulative expected regret}
\begin{align}\label{eq:regretqcb}
\text{Regret}_T^{\mathcal{\gamma,C},\pi } = \sum_{t=1}^T \EX_{\gamma,\mathcal{C},\pi}\left[ \max_{\rho_i\in\gamma}\Tr(\rho_iO_{c_t}) - R_t \right].
\end{align}
For a given action $a\in \mathcal{A}$ and context $O_c \in \mathcal{C}$ the \textit{sub-optimality gap} is defined as
\begin{align}\label{eq:suboptimalitygap}
     \Delta_{a,O_c}=\max_{i\in \mathcal{A}}\Tr(\rho_{i}O_c) - \Tr(\rho_{a}O_c).
\end{align}
Note that the learner could try to learn the distribution of contexts $P_\mathcal{C}$, however, this will not make a difference in minimizing the regret. The strategy of the learner has to be able to learn the relevant part of the unknown states $\lbrace \rho_a \rbrace_{a=1}^{k}$ that depend on the context set and at the same time balance the tradeoff between exploration and exploitation.
We note that it is straightforward to generalize the above setting to continuous sets of contexts $\mathcal{C}$. In order to do that we need a well-defined probability distribution $P_\mathcal{C} (O) dO$ over the context set $\mathcal{C}$.

\section{Lower bound}\label{sec:lowerbound}

In this section, we derive a lower bound for the cumulative expected regret by finding a QCB that is hard to learn for any strategy. Our regret lower bound proof for the QCB model relies on a reduction to a classical multi-armed stochastic bandit given in Theorem 5.1 in~\cite{nonstochasticbandit}. Now we briefly review the multi-armed stochastic bandit problem. 

The \textit{multi-armed stochastic bandit} problem is defined by a discrete set of probability distributions $\nu = (P_a : a\in [k] )$ that is called the environment and $\mu_i$ is the mean of the probability distribution $P_i$ for $i\in [k]$. The learner interacts sequentially with the bandit selecting at each round $t\in[T]$ an action $a\in[k]$ and sampling a reward $R_t$ distributed accordingly to $P_a$. The expected cumulative regret is defined as
\begin{align}
    \text{Regret}_T^{\nu,\pi } = \sum_{t=1}^T \max_{a\in[k]} \mu_a - \EX_{\nu,\pi} [R_t] ,
\end{align}
where $\pi$ and $\EX_{\nu,\pi}$ are both defined analogously from the definitions of the previous section accordingly to this model. It is important to remark that in this setting the actions are independent, meaning that when the learner samples from one action then it cannot use this information to learn about other actions.

Using the above model we describe the multi-armed stochastic bandit studied in Theorem 5.1 in~\cite{nonstochasticbandit}  The bandit is constructed defining an environment $\nu = (P_a : a\in [k] )$ for $k\geq 2$ such that $P_a$ are Bernoulli distributions for all $a\in[k]$ with outcomes $\lbrace l_1,l_2 \rbrace$. Then we set the distributions as follows: we choose an index $i\in[k]$ uniformly at random and assign $P_i (R=l_1) = \frac{1+\Delta}{2}$ for some $\Delta >0$ and $P_a (R = l_1 ) = \frac{1}{2}$ for $a\neq i$.  Thus, there is a unique best action corresponding to $a=i$. Then choosing $\Delta = \small \epsilon\smash{\sqrt{\frac{k}{n}}} \normalsize$ for some small positive constant $\epsilon$, for $n\geq k$ the expected regret for any strategy will scale as 
\begin{align}\label{eq:indep_reg}
\text{Regret}_T^{\nu,\pi } = \Omega (\sqrt{kT} ).
\end{align}

\begin{theorem}
Consider a quantum contextual bandit with underlying dimension $d = 2^n$ and $n\in\mathbb{N}$, context size $c\geq 1$ and $k\geq 2$ actions. Then, for any strategy $\pi$, there exists a context set $\mathcal{C}$, $|\mathcal{C}|=c$, a probability distribution over the context set $\mathcal{C}$ $P_\mathcal{C}$ and an environment $\gamma \in \mathcal{S}_d$ such that for the QCB defined by $(\mathcal{C},\gamma )$ the expected cumulative regret will scale as
\begin{align}\label{eq:qcblowerbound}
\textup{Regret}_T^{\gamma, \mathcal{C},\pi } = \Omega \left(\sqrt{kT}\cdot\min \left\lbrace d,\sqrt{c} \right\rbrace\right),
\end{align}
for $T\geq  k\min \lbrace c,d^2 \rbrace$.
\end{theorem}

\begin{proof}
We use a similar technique to~\cite{abe2003reinforcement,chu} in order to analyze the regret by dividing the problem into subsets of independent rounds. We start dividing the $T$ rounds in $c' = \min \lbrace c,d^2-1 \rbrace$ groups of $T' = \lfloor \frac{T}{c'} \rfloor$ elements. We say that time step $t$ belongs to group $s$ if $\lfloor \frac{t}{T'} \rfloor = s$. We construct a context set $\mathcal{C}$ by picking a set of $c'$ distinct Pauli observables (which is possible since the maximum number of independent Pauli observables is $d^2-1 \geq c'$), so $ \mathcal{C} = \lbrace \sigma_i \rbrace_{i=1}^{c'}$. Recall that a Pauli observable is a $n$-fold tensor product of the 2 × 2 Pauli matrices, thus the reward will be a binary outcome $r_t\in\lbrace -1,1\rbrace$. Then, the context distribution works as follows: at each group $s$ of rounds the learner will receive a different context $\sigma_s\in\mathcal{C}$, so at group $s$ the learner only receives $\sigma_s$. 

We want to build an environment such that for each group of rounds $s\in[m]$ all probability distributions are uniform except one that is slightly perturbed. We associate each Pauli observable $\sigma_i$ to one  unique action $a\in [k]$, and we do this association uniformly at random (each action can be associated with more than 1 Pauli observable). Then each action $a\in[k]$ will have $\lbrace{ \sigma_{a,1},...,\sigma_{a,n_a} \rbrace}$ associated Paulis observables and we can construct the following environment $\gamma = \lbrace \rho_a \rbrace_{a=1}^{k}$ where
\begin{equation}
\rho_a = \frac{I}{d} + \sum_{j=1}^{n_a} \frac{\Delta}{d}\sigma_{a,j},
\end{equation}
$n_a \in \left\lbrace 0,1,...,d^2 -1 \right\rbrace$, $\sum_{a=1}^k n_a = c'$ and $\Delta$ is some positive constant. For every group $s\in[m]$ the learner will receive a fixed context $\sigma_s \in \mathcal{C}$ and there will be a unique action $a'$ with $P_{\rho_a'}(1 | A_t =a',s) =\frac{1}{2}+\frac{\Delta}{2}$ (probability of obtaining $+1$) and the rest $a\neq a'$ will have $P_{\rho_a}(1 | a,s) = \frac{1}{2}$ (uniform distributions). Thus, using that the contexts are independent ($\Tr (\sigma_i\sigma_{j})$ for $i\neq j$) we can apply~\eqref{eq:indep_reg} independently to every group $s$ and we obtain a regret lower bound $ \Omega ( \sqrt{T'k} ) =\Omega ( \smash{\sqrt{\frac{Tk}{c'}}} ) $. Note that in order to apply~\eqref{eq:indep_reg} we need $T' \geq k$ or equivalently $T\geq c'k$. Thus, summing all the $m$ groups we obtain the total regret scales as,
\begin{align}
\text{Regret}_T^{\gamma, \mathcal{C},\pi }  = \Omega \left( c' \sqrt{\frac{Tk}{c'}} \right) = \Omega \left( \sqrt{kT}\cdot\min \left\lbrace d,\sqrt{c} \right\rbrace \right).
\end{align}
\end{proof}

\section{Algorithm}\label{sec:algorithm}
In this section, we review the linear model of multi-armed stochastic bandits and one of the main classical strategies that can be used to minimize regret in this model and also in the QCB model.
\subsection{Linear disjoint single context bandits and QCB}\label{subsec: linear disjoint}
The classical setting that matches our problem is commonly referred to as linear contextual bandits~\cite{chu} although it has received other names depending on the specific setting such as~linear disjoint model~\cite{contextual_new_article_recommendation} or associative bandits~\cite{aurer}. The setting that we are interested uses discrete action sets and optimal algorithms are based on upper confidence bounds (\textsf{UCB}). While these algorithms use the "principle of optimism in the face of uncertainty" there are other approaches like a Thompson sampling~\cite{thompsoncontextual} algorithm but they are not optimal for discrete action sets. We use the contextual linear disjoint bandit model from~\cite{contextual_new_article_recommendation} where each action $a\in [k]$ has an associated unknown parameter $\theta_a\in\mathbb{R}^d$ and at each round $t$ the learner receives a context vector $\mathbf{c}_{t,a}\in \mathbb{R}^d$ for each actions. Then after selecting an action $a\in [k]$ the sampled reward is
\begin{align}\label{eq:contextualproblem}
 R_t = \boldsymbol{\theta}^\top_a \mathbf{c}_{t,a} + \eta_t,
\end{align} 
where $\eta_t$ is some bounded subgaussian~\footnote{A random variable $X$ is called $\sigma$-subgaussian if for al $\mu \in \mathbb{R}$ we have $\EX \left[ \mu X \right] \leq \exp \left(\sigma^2 \mu^2 / 2 \right)$} noise such that $\EX [R_{t} |A_t = a] =  \boldsymbol{\theta}_a \cdot \bold{c}_{t,a}$.

In order to map the above setting to the $d$-dimensional QCB model $(\gamma,\mathcal{C})$ it suffices to consider a vector parametrization (similarly done for the MAQB~\cite{lumbreras21bandit}). We choose a set $\lbrace \sigma_i \rbrace_{i=1}^{d^2}$ of independent Hermitian matrices and parametrize any $\rho_a \in \gamma$ and $O_l \in\mathcal{C}$ as
\begin{align}\label{eq:qcbvecparametrization}
    \rho_a = \sum_{i = 1}^{d^2}\theta_{a,i} \sigma_i , \quad O_l = \sum_{i=1}^{d^2} c_{l,i} \sigma_i ,
\end{align}
where $\theta_{a,i} = \Tr (\rho_a \sigma_i ) $ and $c_{l,i} = \Tr (O_c \sigma_i) $ and we define the vectors $\boldsymbol{\theta}_a = (\theta_{a,i})_{i=1}^{d^2} \in \mathbb{R}^{d^2}$ and  $\mathbf{c}_l  = (c_{a,i})_{i=1}^{d^2}  \in \mathbb{R}^{d^2}$. Then we note that for the QCB model the rewards will be given by~\eqref{eq:contextualproblem} with the restriction that since we only receive one observable at each round then the context vector is constant among all actions. Thus, in our model, the rewards have the following expression
\begin{align}\label{eq:singlecontextbandit}
    R_t = \boldsymbol{\theta}_a^\top \bold{c}_{t} + \eta_t.
\end{align}
We denote this classical model as \textit{linear disjoint single context bandits}. In order to make clear when the classical real vectors parametrize an action $\rho_a\in\gamma$ or context $O_l \in \mathcal{C}$~\eqref{eq:qcbvecparametrization} we will use the notation $\boldsymbol{\theta}_{\rho_a}$ and $\mathbf{c}_{O_l}$  respect to the standard Pauli basis.

\subsection{Linear Upper Confidence Bound algorithm}

Now we discuss the main strategy for the linear disjoint single context model~\eqref{eq:singlecontextbandit} that is the \textsf{LinUCB} (linear upper confidence bound) algorithm~\cite{aurer,chu,lin1,lin2,lin3}. We describe the procedure of \textsf{LinUCB} for selecting an action and we leave for the next section a complete description of the algorithm for the QCB setting.

At each time step $t$, given the previous rewards $R_1,...,R_{t-1} \in \mathbb{R}$, selected actions $a_1,...,a_{t-1} \in [k]$ and observed contexts $\mathbf{c}_1 ,..., \mathbf{c}_t \in \mathbb{R}^d$ the \textsf{LinUCB} algorithm builds the \textit{regularized least squares estimator} for each unknown parameter $\boldsymbol{\theta}_a$ that have the following expression
\begin{align}
\tilde{\boldsymbol{\theta}}_{t,a}  = V_{t,a}^{-1} \sum_{s=1}^{t-1} R_s \textbf{c}_{s} \mathbb{I} \lbrace a_t = a \rbrace,
\end{align}
where $V_{t,a} = I + \sum_{s=1}^{t-1} \textbf{c}_{s} \textbf{c}^\top_{s} \mathbb{I} \lbrace a_t = a \rbrace$.
Then \textsf{LinUCB} selects the following action according to
\begin{align}
a_{t+1} = \argmax_{a\in [k]} \tilde{\boldsymbol{\theta}}^\top_{t,a} \mathbf{c}_{t} + \alpha \sqrt{  \mathbf{c}^\top_{t} {V^{-1}_{t}} \mathbf{c}_t } ,
\end{align}
where $\alpha > 0$ is a constant that controls the width of the confidence region on the direction of $\textbf{c}_{t,a}$. The idea behind this selection is to use an overestimate of the unknown expected value using an upper confidence bound. This is the principle behind \textsf{UCB}~\cite{firstUCB} which is the main algorithm that gives rise to this class of optimistic strategies.
The value of the constant $\alpha$ is chosen depending on the structure of the action set. In the next section, we will discuss the appropriate choice of $\alpha$ for our setting.
We note that this strategy can also be applied in an adversarial approach where the context is chosen by an adversary instead of sampled from some probability distribution. 

The above procedure is shown to be sufficient for practical applications~\cite{contextual_new_article_recommendation} but the algorithms achieving the optimal regret bound are \textsf{SupLinRel}~\cite{aurer} and \textsf{BaseLinUCB}~\cite{chu}. They use a phase elimination technique that consists of each round playing only with actions that are highly rewarding but still the main subroutine for selecting the actions is \textsf{LinUCB}. This technique is not the most practical for applications but it was introduced in order to derive rigorous regret upper bounds. For these strategies if we apply it to a $d$-dimensional QCB bandit $(\gamma,\mathcal{C})$ they achieve the almost optimal regret bound of
\begin{align}\label{eq:sublinear}
\text{Regret}_T^{\gamma, \mathcal{C},\pi } = O \left( d\sqrt{kT\ln^3(T^2\log (T))} \right).
\end{align}
The above bound comes from~\cite{chu} and it is adapted to our setting \footnote{Our model uses a different unknown parameter $\boldsymbol{\theta}_a \in \mathbb{R}^d$ for each action $a\in [k]$. This model can be easily adapted to settings where they assume only one unknown parameter shared by all actions if we enlarge the vector space and define $\boldsymbol{\theta} = (\theta_1,..., \theta_k ) \in\mathbb{R}^{dk}$ as the unknown parameter.} using the vector parametrization~\eqref{eq:qcbvecparametrization}. Their regret analysis works under the normalization assumptions $\| \boldsymbol{\theta} \|_2 \leq 1$, $\| \mathbf{c}_t \|_2 \leq 1$ and the choice of $\alpha = \sqrt{\frac{1}{2} \ln (2T^2 k ) } $.  We note that it matches our lower bound~\eqref{eq:qcblowerbound} except for the logarithmic terms.

\section{Low energy quantum state recommender system}\label{sec:lowenergy}

In this section, we describe how the QCB framework can be adapted for a recommender system for low-energy quantum states.
We consider a setting where the learner is given optimization problems in an online fashion and is able to encode these problems into Hamiltonians and also has access to a set of unknown preparations of (mixed) quantum states that they want to use in order to solve these optimization problems. The task is broken into several rounds; at every round, they receive an optimization problem and are required to choose the state that they will use for that problem. As a recommendation rule, we use the state with the lowest energy with respect to the Hamiltonian where the optimization problem is encoded. We denote this problem as the \textit{low energy quantum state recommendation problem}. We note that our model focuses on the recommendation following the mentioned rule. After selecting the state the learner will use it for the optimization problem (for example the initial ansatz state of a variational quantum eigensolver), but that is a separate task.
When a learner chooses an action, they must perform an energy measurement using the given Hamiltonian on the state corresponding to the chosen action. Then the measurement outcomes are used to model rewards, and their objective is to maximize the expected cumulative reward, i.e, the expectation on the sum of the measurement outcomes over all the rounds played. These measurements can be done fairly simply. Any Hamiltonian can be written as a linear combination of Pauli observables. Now by measuring each of these Pauli observables (since these measurements are conceivable, \cite{peruzzo2014variational}), and taking the appropriately weighted sum of the measurement outcomes, we can simulate such a measurement. The QCB framework naturally lends itself to this model, where the Hamiltonians are the contexts, and the set of states that can be prepared reliably serve as the actions.

In this paper we study some important families of Hamiltonians --- specifically, the Ising and a generalized cluster model from~\cite{verresen2017one}, which are linear combinations of Pauli observables with nearest-neighbor interactions and for $n$ qubits can be written as 
\begin{align}\label{eq:ising}
&H_\text{ising}(h)=\sum_{i=1}^{n}(Z_iZ_{i+1}+hX_i),\\
 \label{eq:cluster}   &H_\text{cluster}(j_1,j_2)=\sum_{i=1}^{n} (Z_i-j_iX_iX_{i+1}-j_2X_{i-1}Z_iX_{i+1}),
\end{align}
where $h,j_1,j_2 \in \mathbb{R}$. In the Ising model, $h$ corresponds to the external magnetic field. Specifically, we consider QCB with the following context sets

\begin{align}\label{eq:contexts}
    \mathcal{C}_{\text{Ising}} = \left\lbrace H_\text{ising}(h) : h\in \mathbb{R} \right\rbrace, \quad  \mathcal{C}_{\text{cluster}} = \left\lbrace  H_\text{cluster}(j_1,j_2) : j_1,j_2\in \mathbb{R} \right\rbrace .
\end{align}
   
Important families of Hamiltonians like the models discussed above show translation-invariance and are spanned by Pauli observables showing nearest-neighbor interactions, and as a result, span a low dimensional subspace. We illustrate the scheme described above through the example of the Ising Model contexts. The Pauli observables that need to be measured are $\lbrace X_i\rbrace_{ i \in [n]}$ and $\lbrace Z_iZ_{i+1} \rbrace_{i \in [n]}$ . These observables have 2 possible measurement outcomes, -1 and 1, and by the reward distribution of a Pauli observable $M$ given by Born's rule~\eqref{eq:reward_distribution} on a quantum state $\rho$, the reward can be modeled as
\begin{align}
    R_{M,\rho}= 2\text{Bern}\left(\frac{\Tr(M\rho)+1}{2}\right)-1 , 
\end{align}
where $\text{Bern}(x) \in \lbrace 0 , 1 \rbrace$ is a random variable with Bernoulli distribution with mean $x\in [0,1]$. By performing such a measurement for all the Pauli observables and adding the rewards, the reward for $\mathcal{C}_\text{Ising}$ is
\begin{align}
    R_{\text{Ising}}=-h\sum_{M \in {X_i,i \in [n]}}R_{M,\rho} -\sum_{M' \in {Z_iZ_{i+1},i \in [n]}}  R_{M',\rho},
\end{align}
where we took the negative of the sum of the measurements because we are interested in a recommender system for the lowest energy state. A similar formulation applies to the QCB with generalized cluster Hamiltonian contexts.\\
In the rest of this section, we illustrate a modified \textsf{LinUCB} algorithm for the QCB setting. Then we implement this recommender system where the contexts are Hamiltonians belonging to the Ising and a generalized cluster models~\eqref{eq:contexts}, and demonstrate our numerical analysis of the performance of the algorithm by studying the expected regret. We also demonstrate that depending on the action set, the algorithm is able to approximately identify the phases of the context Hamiltonians. 

\subsection{Gram-Schmidt method}\label{subsec:gram}

Similarly to the task of shadow tomography~\cite{aaronson2018shadow} and classical shadows~\cite{huang2020predicting}, we do not need to reconstruct the full quantum states since the algorithm has only to predict the trace between the contexts and the unknown quantum states. Thus, the \textsf{LinUCB} algorithm has only to store the relevant part of the estimators for this computation. As the measurement statistics depend only on the coefficient corresponding to the Pauli observables spanning the observables in the context set, only those Pauli observables in the expansion of the estimators are relevant. This means that our algorithm can operate in a space with a smaller dimension than the entire spaces spanned by $n$ qubits, which has a dimension that is exponential on the number of qubits.

In order to exploit this property to improve the space complexity of the \textsf{LinUCB} algorithm, we use the Gram-Schmidt procedure in the following way. At any round, a basis for the vector parameterizations (as shown in~\eqref{eq:qcbvecparametrization}) of all the previously received contexts is stored. If the incoming  vector parameterization of the context is not spanned by this basis, the component of the vector orthogonal to the space spanned by this set is found by a Gram-Schmidt orthonormalization-like process, and this component is added to the set, after normalization. Therefore, at any round, there will be a list of orthonormal vectors that span the subspace of all the vector parameterizations of the contexts  received so far, and the size of the list will be equal to the dimension of the subspace, which we call \textit{effective dimension}, i.e, \\
\begin{align}\label{eq:d_eff}
    d_{\text{eff},t}=\text{dim}(\{O_{c_t} \in \mathcal{C}\}: t \in [T]).
\end{align}
From now on we will omit the subscript for the time step $t$ and simply denote the effective dimension as $d_\text{eff}$. Instead of feeding the context vectors directly, for any incoming context vector, we construct $d_\text{eff}$-dimensional vectors, whose $i^{th}$ term is the inner product of the context vector and the $i^{th}$ basis vector. In case the incoming vector is not spanned by the basis, we first update the list by a Gram-Schmidt procedure (which will result in an addition of another orthonormal vector to the list, and an increase in $d_\text{eff}$ by 1), and then construct a $d_\text{eff}$-dimensional vector as described before.
 This vector is fed to the \textsf{LinUCB} algorithm. The Gram-Schmidt procedure is stated in Algorithm ~\ref{alg:grahm-schmidt} and the modified \textsf{LinUCB} algorithm is stated explicitly in Algorithm ~\ref{alg:LinUCBgram}.
The efficiency of this method is well illustrated in the case where all the contexts are local Hamiltonians. As an example, we discuss the case of generalised cluster Hamiltonians. Note that the space complexity of the standard QCB framework is $O(kd^2)$, where k is the number of actions, and d is the dimension of the vector parameterizations of the contexts. In the standard \textsf{LinUCB} technique, the context vectors $\mathbf{c_t},t \in [T]$ would be $4^n$-dimensional, where n is the number of qubits the Hamiltonian acts on, in which case the space complexity of the algorithm is $O(k4^{2n})$ . 
In our studies, the contexts are Ising Hamiltonians and a generalised cluster Hamiltonian~\eqref{eq:contexts} with $d_\text{eff}\leq 2$ and $d_\text{eff}\leq 3$ respectively. Since the vectors fed into the modified \textsf{LinUCB} is $d_\text{eff}$-dimensional, the space complexity is $O(kd_\text{eff}^2)$, i.e, $O(4k)$ and $O(9k)$ respectively. 

\begin{algorithm}[H]
	\caption{Gram-Schmidt Algorithm ($\text{Gram}(\mathbf{c},{V_{a}},\mathbf{b}_a,\text{CBasis})$)}
	\label{alg:grahm-schmidt}
	\begin{algorithmic}
	    \State Input $[\mathbf{c},\{V_{a}\}_{a \in \mathcal{A}}, \{\mathbf{b}_a\}_{a \in \mathcal{A}},\text{CBasis}]$
		    \For {\textbf{v} in CBasis}  
    		    \State $\textbf{c} \leftarrow \textbf{c}- (\mathbf{v}^{\top} \textbf{c})\textbf{v}$
    		    \State $\mathbf{v}_\text{ct}=\mathbf{v}_\text{ct} \oplus (\mathbf{v}^{\top} \textbf{c})$

    	       \EndFor
            \If {$\textbf{c}!=\textbf{0}$}
                \State $\mathbf{v}_\text{ct}=\mathbf{v}_\text{ct} \oplus \|\mathbf{c}\|_2$
                \State Add $\textbf{c}/\|\textbf{c}\|_2$ to CBasis
    		    \For {$a=1,2,\ldots, K$}
    		        \State Set $V_a= V_a \oplus I_{1}$, $\mathbf{b}_a= \mathbf{b}_a \oplus \mathbf{0}_{1}$
    		    \EndFor 
                \EndIf
            \State \textbf{Return} $\left[\mathbf{c'},\{V_{a}\}_{a \in \mathcal{A}}, \{\mathbf{b}_a\}_{a \in \mathcal{A}},\text{CBasis}\right]$	
	\end{algorithmic} 
 \end{algorithm}

\begin{algorithm}[H]
	\caption{\textsf{LinUCB} with Gram-Schmidt} 
	\label{alg:LinUCBgram}
	\begin{algorithmic}[1]
        \State Input $\alpha\in\mathbb{R}$ 
        \State Set CBasis = $\left[ \,\, \right]$
        \State Set $V_a= \mathbf{1} ,\mathbf{b}_a=\mathbf{0},  \forall a \in \mathcal{A}$
		\For {$t=1,2,\ldots$}
            \State $\left[\mathbf{c'}_{O_t},\{V_{a}\}_{a \in \mathcal{A}}, \{\mathbf{b}_a\}_{a \in \mathcal{A}},,\text{CBasis}\right]\leftarrow \text{Gram}(\mathbf{c}_{O_t},\{V_{a}\}_{a \in \mathcal{A}}, \{\mathbf{b}_a\}_{a \in \mathcal{A}},,\text{CBasis})$
            \For {$a \in \mathcal{A}$}
                \State $\tilde{\boldsymbol{\theta}}_{\rho_a} \leftarrow V^{-1}_{a}\mathbf{b}_a$
                \State $p_{t,a} \leftarrow \tilde{\boldsymbol{\theta}}_{\rho_a}\mathbf{c'}_{O_t}+\alpha\sqrt{\mathbf{c}^{'\top}_{O_t} V^{-1}_{a}\mathbf{c'}_{O_t}}$
                \EndFor
    		\State Choose action $\mathbf{a_t} = \argmax_{a \in \mathcal{A}}p_{t,a}$;
    		\State Measure state $\rho_{a_t}$ with $O_{c_t}$ and observe reward $R_{O_t}$
    		\State Set $V_{a_t} \leftarrow V_{a_t}+\mathbf{c'}_{O_t}\mathbf{c'}^{'\top}_{O_t}$
    		\State Set $\mathbf{b}_{a_t} \leftarrow \mathbf{b}_{a_t}+R_{O_t}\mathbf{c'}_{O_t}$
			
			\EndFor
	
	\end{algorithmic} 
\end{algorithm}
\subsection{Phase classifier}

In order to implement the numerical simulations we need to choose the environments for the QCB with context sets $\mathcal{C}_\text{ising}$ and $\mathcal{C}_\text{cluster}$. Elements of both context sets are parameterized by tunable parameters. We study the performance of the recommender system by choosing a context probability distribution that is uniform on these parameters.  Then we chose the actions 
 as ground states of Hamiltonians that corresponded to the limiting cases (in terms of the parameters) of these models. In order to study the performance of our strategy apart from the expected regret~\eqref{eq:regretqcb} we want to observe how the actions are chosen. For every action, we maintained a set, which contained all the Hamiltonians for which that action was chosen. We observed that almost all the elements in each of these sets belonged to the same phase of the Hamiltonian models. 
 
 In order to study the performance of the algorithm in this respect, we define the \textit{classifier regret} as
\begin{align}
   \text{ClassifierRegret}_{T}^{\mathcal{\gamma,C},\pi }=\sum_{t=0}^{T-1}\mathbb{I}\left[a_t \neq a_\text{optimal,t} \right],
\end{align}
where $a_\text{optimal,t}=\argmax_{a \in [k]}\tr{O_{t}\rho_a}$, and $O_t \in \mathcal{C}$ is the context observable received in $t^\text{th}$ round.
Note that the above classifier regret is not guaranteed to be sublinear, like expected regret is ~\eqref{eq:sublinear} for the \textsf{LinUCB} strategy. This can be understood intuitively: consider a scenario where the bandit picks an actions with a small sub-optimality gap~\eqref{eq:suboptimalitygap}; then the linear regret will increase by a very small amount, the classifier regret will increase by one unit, as all misclassifications have equal contribution to regret. These, however, are theoretic worst-case scenarios, and this classifier regret is useful to study the performance of the algorithm in practice in our settings.

\subsection{Numerical simulations}\label{sec:simulations}

 Before we move into the specific cases, we note the importance of the choice of $\alpha$ in Algorithm~\ref{alg:LinUCBgram}. While the theoretical analysis of the \textsf{LinUCB} algorithm depends on the choice of $\alpha$ , in practice one can tune this value to observe a better performance. We primarily use  the $\alpha$ described in~\cite{lattimore2020bandit} (Chapter 19) given by 
\begin{align}
    \alpha_t= m + \sqrt{2\log\left( 
\frac{1}{\delta} \right)+d\log\left(1+\frac{tL^2}{d} \right)}.
\end{align}
Here, $L$ and $m$ are upper bounds on the 2-norm of the action vectors and unknown parameter respectively,  $d$ is the dimension and $\delta$ is once more a probability of failure.

Finally, while we study the performance of our algorithm in our simulations with estimates of expected regret and expected classifier regret, it is important to note that in an experimental setup, the learner will only be able to measure the cumulative reward at every round. However, since these are simulations, we are able to study the regret as well, as they are standard metrics to gauge the performance of the algorithms. In the next subsection we discuss our simulations of the QCB bandit $(\gamma,\mathcal{C}_\text{cluster})$ model and later, the QCB bandit $(\gamma,\mathcal{C}_\text{Ising})$ is discussed in the Appendix \ref{sec:Appendix}.

\subsubsection{Generalised Cluster Model}
We study the performance of the recommender system for the QCB bandit $(\gamma,\mathcal{C}_\text{cluster})$, where the generalised cluster Hamiltonians  ~\cite{verresen2017one}, act on 10 qubits and 100 qubits respectively. We observe that the performance of the algorithm is not affected by the number of qubits, as the effective dimension of the context set remains unchanged,i.e, $d_\text{eff}=3$. We study the expected regret and classifier regret for these two cases, and illustrate the system's performance in finding the phases of the generalised cluster Hamiltonians. This model was also studied in \cite{caro2021generalization}, where they designed a quantum convolutional neural network to classify quantum states across phase transitions. We chose 5 actions corresponding to approximate ground states of Hamiltonians that are the limiting cases of the generalised cluster model; i.e, generalised cluster Hamiltonians with parameters $j_1,j_2$ in ~\eqref{eq:cluster}, $j_1,j_2\rightarrow \lbrace 0,0 \rbrace,\lbrace 0,\infty \rbrace,\lbrace \infty,0 \rbrace,\lbrace 0,-\infty \rbrace$and $\lbrace 0,-\infty \rbrace$. Note that these methods of approximating ground states is only for simulation purposes.
\begin{figure}
    \centering
    \includegraphics[scale=.4]{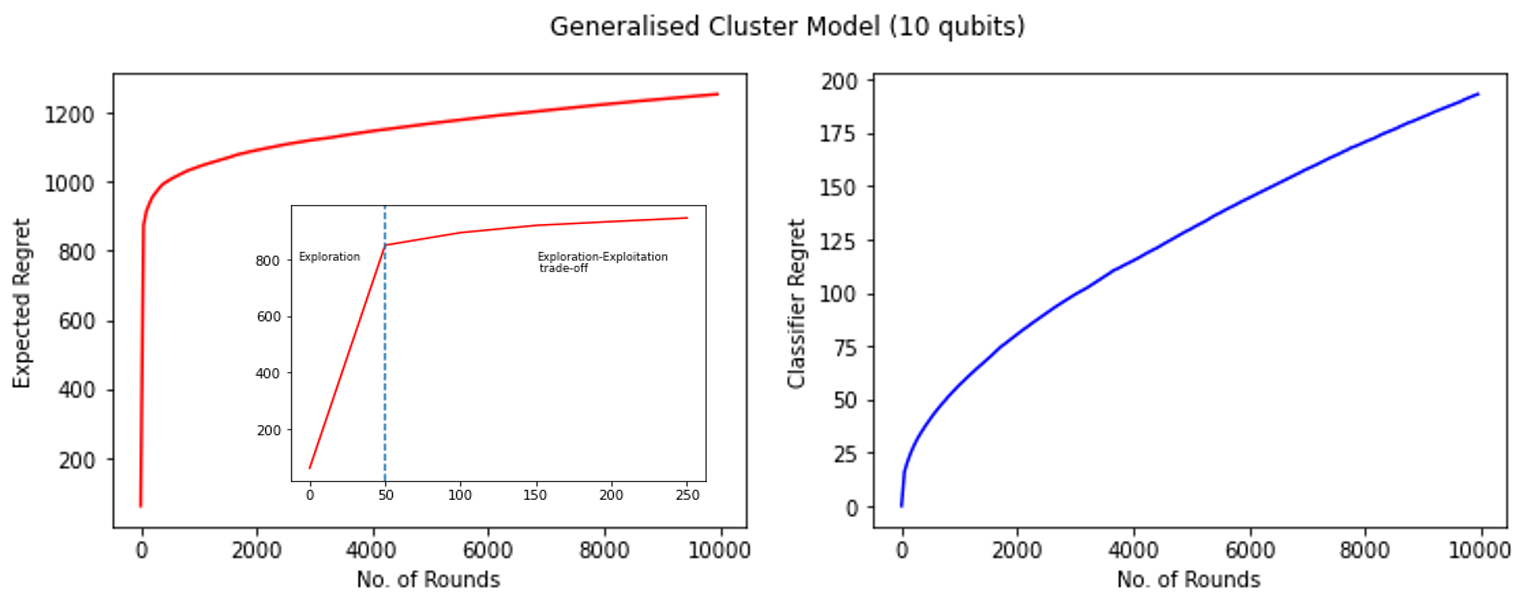}
    \includegraphics[scale=.4]{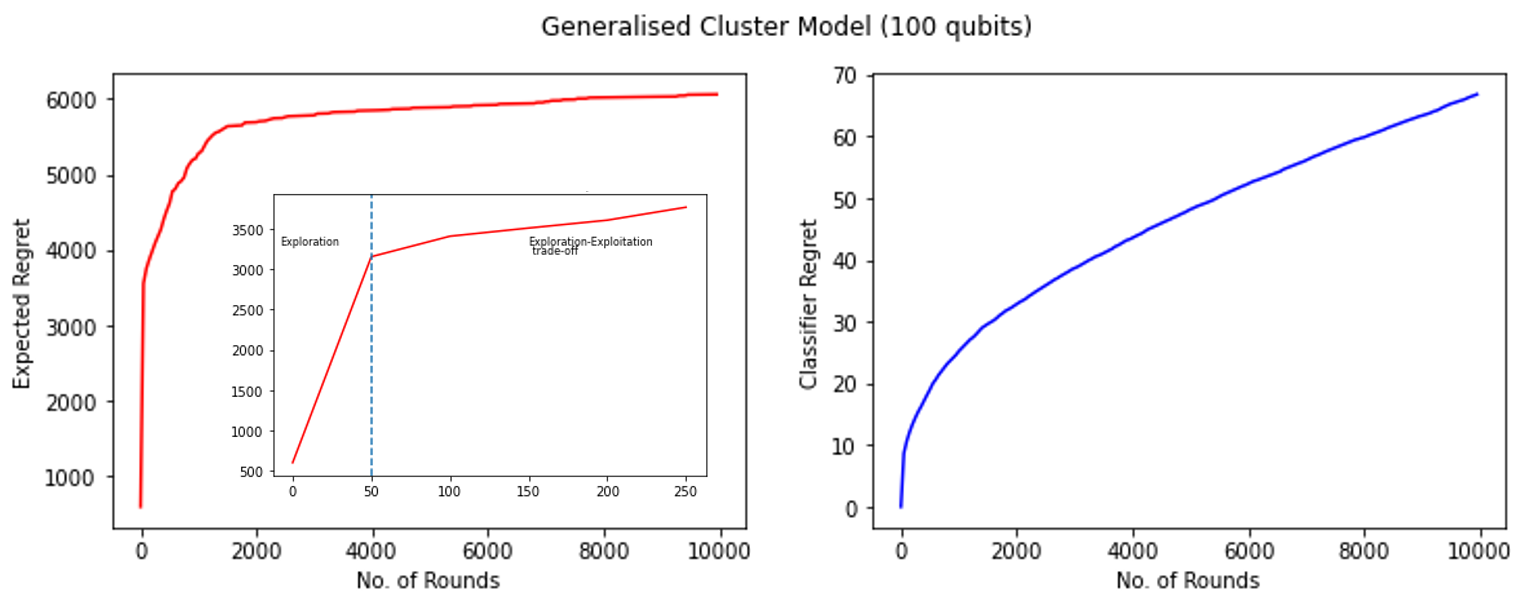}
    \caption{Plots for Regret and Classifier regret for QCB bandit $(\gamma,\mathcal{C})$, where the Hamiltonians in $\mathcal{C}$ are a specific form of generalised cluster models acting on 10 and 100 qubits respectively. The performance is not very different since $d_\text{eff}=3$ ~\eqref{eq:d_eff} for both cases. The action set is chosen to be approx. ground states of some generalised cluster Hamiltonians}
    \label{fig:clusterregret}
\end{figure}

Initially a steep growth in regret is observed, followed by sudden slower pace. On looking closely, in the plot below, we find that the regret indeed continues to grow, albeit at a slower pace. This was be explained by observing that the sub-optimality gap of the second-best action is quite small in comparison to the sub-optimality gaps of the rest of the actions. 
Initially the \textsf{LinUCB} algorithm does not have enough information about the unknown parameters and has to play all actions resulting in an exploration phase. However, at some point the bandit recognizes the "bad" actions, and plays either the best action or the action with a small sub-optimality gap most of the time - this is when the bandit has begun to balance exploration and exploitation. This is illustrated by observing the growth of the regret before and after the first 50 rounds in the insets of Fig. \ref{fig:clusterregret}.

In the beginning of this subsection, we had mentioned that the recommendation system picks the same action for context Hamiltonians belonging to the same phase.
 We illustrate this in Figure \ref{fig:phaseCluster}. In the scatter plot, when a context generalised cluster Hamiltonian is received, a dot is plotted with the x-axis  and y- axis coordinates corresponding to its parameters $j_1,j_2$ respectively. Depending on the action picked by the algorithm, we associate a color to the dot. The resultant plot is similar (but not exact) to the phase diagram of the generalised cluster Hamiltonian.
\begin{figure}[H]
    \centering
    \includegraphics[scale=.5]{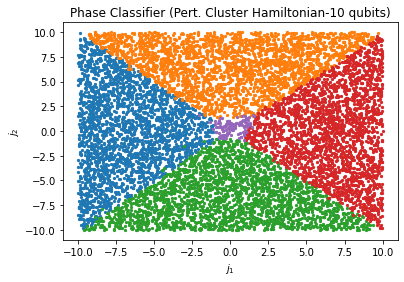}
    \includegraphics[scale=.5]{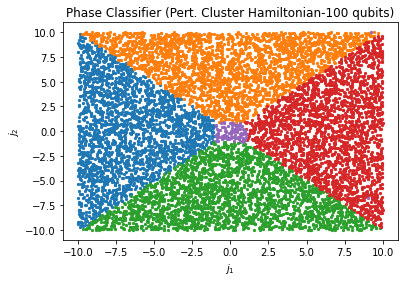}
    \caption{These plots illustrate how the recommender system identifies the phases of the generalised cluster Hamiltonian. The x and y-axis represent the coupling coefficients of the generalised cluster Hamiltonian received as context. Like the Ising Model simulations, we associate a color to each action. For any context $H_\text{cluster}(j_1,j_2)$  corresponding to any of the T rounds, one of these actions is picked by the algorithm. We plot the corresponding colored dot (blue for ground state of $H_\text{cluster}(-\infty,0)$, orange for $H_\text{cluster}(0,\infty)$, red for $H_\text{cluster}(\infty,0)$, green for $H_\text{cluster}(0,-\infty)$ and purple for $H_\text{cluster}(0,0)$) at the appropriate coordinates, for rounds that follow after the bandit has "learned" the actions, i.e, the growth in regret has slowed down.} 
    \label{fig:phaseCluster}
\end{figure}

\section{Outlook}
This work describes the first steps for recommending quantum data  by implementing the bandit framework in a rigorous fashion for practical scenarios. We provide a recommender system based on the theory of linear contextual bandits and show that the upper and lower-bounds on the expected regret are tight except for logarithmic factors. We also demonstrate its efficiency in practice through simulations. Later, we show how such a system could also be used to recognise phases of Hamiltonians.

We restricted our attention to a model where the expected rewards follow a linear function in terms of the context and the unknown states. While the low energy quantum state recommendation problem uses the outcome of the measurement as a reward, one could think of other recommendation tasks with more complicated reward functions. Non-linear rewards have been studied in the bandit literature and receive the name of structured bandits~\cite{lattimore2014bounded,combes2017minimal,russo2013eluder}. This model could be a natural extension of the QCB for other recommendation tasks where the rewards are not in one-to-one correspondence with measurement outcomes. Going back to the general model, the environment is modeled by a set of unknown quantum processes which in the QCB model we assumed to be a set of stationary unknown quantum states. In a more general scenario, we can consider environments that change with time due to some Hamiltonian evolution or noise interaction with an external environment. In the bandit literature, non-stationary environments were first considered in~\cite{gittins1979bandit,gittins2011multi} where each action was associated with a Markov chain or the restless bandit model~\cite{whittle1988restless} where the Markov chain associated to each action evolves with time. More recently in~\cite{luo2018efficient} they studied a contextual bandit model with non-stationary environments. We expect that recommender systems for quantum data can also be extended to similar settings.

\textbf{Acknowledgements:} This research is supported by the National Research Foundation, Singapore and A*STAR under its CQT
Bridging Grant and the Quantum Engineering Programme grant NRF2021-QEP2-02-P05.

\bibliographystyle{unsrt}
\bibliography{biblio}
\newpage
\section*{Appendix} \label{sec:Appendix}
In this appendix, we discuss the simulations for the QCB bandit $(\gamma,\mathcal{C}_\text{ising})$ setting, in a similar fashion as described in subsection \ref{sec:simulations}.
We study the performance of the recommender system for the QCB bandit $(\gamma,\mathcal{C}_\text{ising})$, where the Ising Hamiltonians act on 10 qubits and 100 qubits respectively. We observe that the performance of the algorithm is not affected by the number of qubits, as the effective dimension of the context set remains unchanged. We study the expected regret and classifier regret for these two cases, and illustrate the system's performance in finding the phases of the Ising Model.  The action set corresponds to the ground state of the 3 limiting cases of the Ising Model i.e, Ising Hamiltonians for parameter $h$ in~\eqref{eq:ising}, $h=0,h\rightarrow -\infty$ and $h \rightarrow \infty$ .
\begin{figure}[H]
    \centering
    \includegraphics[scale=.5]{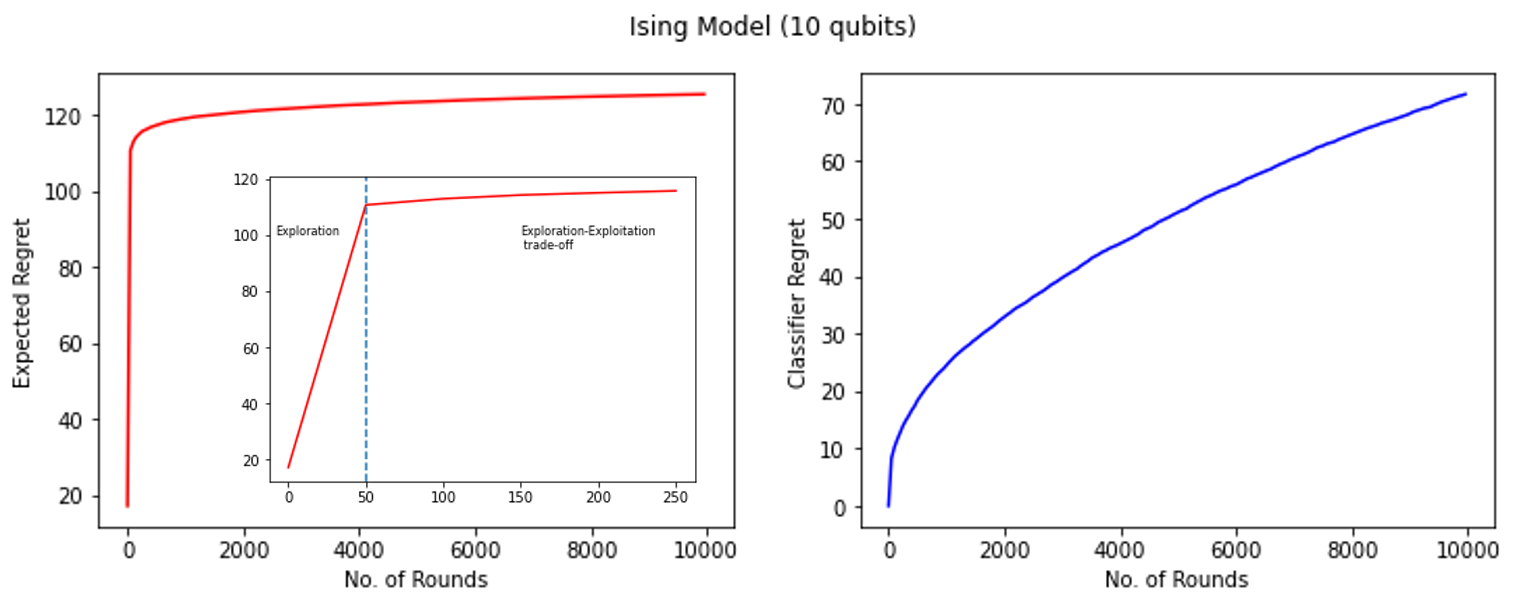}
    \includegraphics[scale=.5]{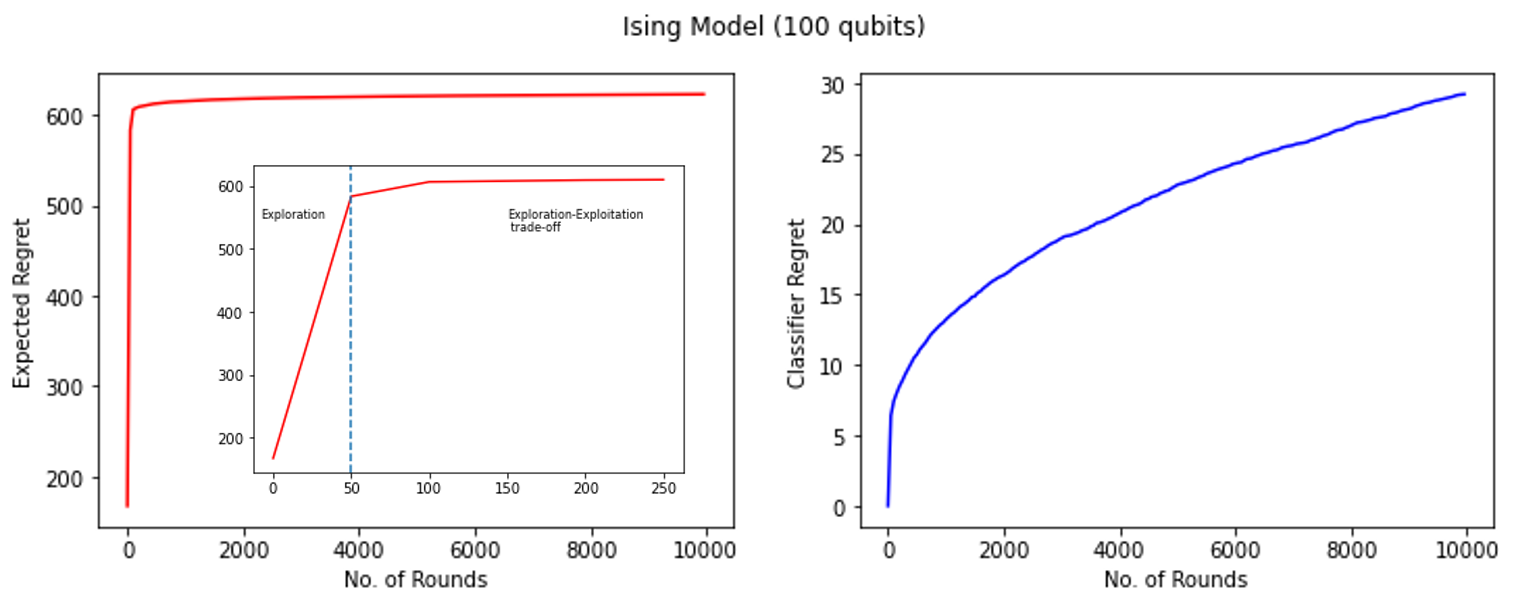}
    \caption{Plots for Regret and Classifier regret for QCB bandit $(\gamma,\mathcal{C})$, where the Hamiltonians in $\mathcal{C}$ are Ising Hamiltonians acting on 10 and 100 qubits respectively. The performance is not very different since $d_\text{eff}=2$ ~\eqref{eq:d_eff} for both cases. The action set is chosen to be approx. ground states of some Ising Hamiltonians}
    \label{fig:iusingregret}
\end{figure}

Once more, like the generalised cluster Model simulations, we observe a distinct "exploration" stage, followed by a "exploration-exploitation trade-off" stage which we again plot separately. This is illustrated by observing the growth of the regret before and after the first 50 rounds shown in the insets of Fig. \ref{fig:iusingregret}.

In the beginning of this subsection, we had mentioned that the recommendation system picks the same action for context Hamiltonians belonging to the same phase.
 We illustrate this in Figure \ref{fig:phaseIsing}. In the scatter plot, when a context Ising Hamiltonian is received, a dot is plotted with the x-axis coordinate corresponding to its parameter. Depending on the action picked by the algorithm, we associate a color to the dot. The resultant plot is very similar to that of the phase diagram of an Ising model. The Ising model is known to have phase transitions at $h=-1,1$, resulting in 3 phases, $(-\infty,-1],[-1,1]$ and $[1,\infty)$, and we observe that different actions were played in each of these ranges.
\begin{figure}[H]
    \centering
    \includegraphics[scale=.5]{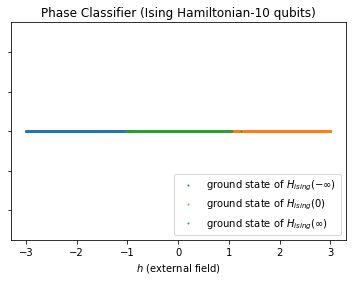}
    \includegraphics[scale=.5]{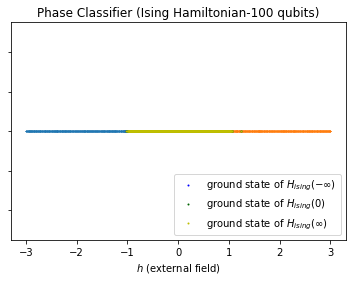}
    \caption{These plots illustrate how the recommender system identifies the phases of the Ising Hamiltonian. The x-axis represents the external field coefficient of the Ising Hamiltonian received as context. The blue, green, or yellow mark indicates that the algorithm plays the $1^\text{st}$,$2^\text{nd}$ or $3^\text{rd}$ action. We plot the corresponding colored dot at the appropriate coordinates, for rounds that follow after the bandit has "learned" the actions, i.e, the growth in regret has slowed down.}
    \label{fig:phaseIsing}
\end{figure}
\end{document}